\documentclass[a4paper, 12pt]{article}

\usepackage[sort&compress]{natbib}
\bibpunct{(}{)}{;}{a}{}{,} 

\usepackage{amsthm, amsmath, amssymb, mathrsfs, multirow, url, subfigure}
\usepackage{graphicx} 
\usepackage{ifthen} 
\usepackage{amsfonts}
\usepackage[usenames]{color}
\usepackage{fullpage}

\theoremstyle{plain} 
\newtheorem{thm}{Theorem}
\newtheorem{cor}{Corollary}

\theoremstyle{definition}

\theoremstyle{remark}

\newtheorem{ex}{Example}
\newtheorem{remark}{Remark}

\newcommand{\prob}{\mathsf{P}} 
\newcommand{\E}{\mathsf{E}}

\newcommand{\bel}{\mathsf{bel}}
\newcommand{\pl}{\mathsf{pl}}

\newcommand{\bin}{{\sf Bin}}
\newcommand{\unif}{{\sf Unif}}
\newcommand{\nm}{{\sf N}}

\newcommand{\gam}{{\sf Gamma}}

\newcommand{\chisq}{{\sf ChiSq}}

\newcommand{\YY}{\mathbb{Y}}
\newcommand{\UU}{\mathbb{U}}

\newcommand{\bigmid}{\; \Bigl\vert \;}

\renewcommand{\L}{\mathcal{L}}
\newcommand{\C}{\mathcal{C}}

\renewcommand{\S}{\mathcal{S}}

\renewcommand{\phi}{\varphi} 
\newcommand{\eps}{\varepsilon}

\newcommand{\avar}{\sigma_\alpha^2}
\newcommand{\evar}{\sigma_\eps^2}

\newcommand{\stgeq}{\geq_{\text{st}}}

\title{On an inferential model construction using \\ generalized associations}
\author{Ryan Martin \\
Department of Mathematics, Statistics, and Computer Science \\
University of Illinois at Chicago \\
\url{rgmartin@uic.edu} 
}
\date{\today}

\begin{document}

\maketitle 

\begin{abstract}    
The inferential model (IM) approach, like fiducial and its generalizations, depends on a representation of the data-generating process.  Here, a particular variation on the IM construction is considered, one based on generalized associations.  The resulting generalized IM is more flexible than the basic IM in that it does not require a complete specification of the data-generating process and is provably valid under mild conditions.  Computation and marginalization strategies are discussed, and two applications of this generalized IM approach are presented.  

\smallskip

\emph{Keywords and phrases:}  Likelihood; marginalization; Monte Carlo; plausibility function; random set; validity.
\end{abstract}

\section{Introduction}
\label{S:intro}  

An advantageous feature of the mainstream approaches to statistical inference is simplicity.  On one hand, likelihood-based approaches, including ``Frasian'' inference \citep[e.g.,][]{reid2003, fraser1990, fraser1991, bn1991, fraser2011} and certain forms of Bayesian inference \citep[e.g.,][]{bernardo1979, mghosh2011, bergerbernardosun2009, berger.bernardo.sun.2015}, are simple in the sense that the calculations relevant to data analysis are largely (or completely) determined by the posited sampling model.  On the other hand, frequentist approaches are also simple because the ``do whatever works well'' viewpoint is extremely flexible.  This is in sharp contrast with fiducial inference \citep{fisher1973, dawidstone1982, barnard1995, taraldsen.lindqvist.2013}, its generalizations \citep{hannig2009, hannig.review.2015}, and the recently proposed inferential model (IM) framework \citep{imbasics, imcond, immarg, imbook}, which appear to be not-so-simple in the sense that their construction depends on something more than the data and sampling model.  In particular, the fiducial and IM construction begins with a specific representation of the data-generating mechanism, one that determines but is not determined by the sampling model.  This data-generating mechanism identifies an auxiliary variable, or pivotal quantity, that controls the random variation in the observable data.  A familiar example of this kind is the regression model, $Y = X\beta + \sigma \eps$, where the random ``$\eps$'' part controls the variation of the response $Y$ around the deterministic ``$X\beta$'' part.  That the fiducial and IM solutions depend on the choice of the data-generating mechanism may be seen as a shortcoming of these approaches.   

One approach to deal with the choice of data-generating mechanism is to find one that is ``best'' in some sense; for example, \citet{majumdar.hannig.2015} compare different data-generating mechanisms using higher-order asymptotics in the fiducial context.  Since defining and identifying the ``best'' is difficult, I want to take a different approach.  In this paper, building on \citet[Ch.~11]{imbook}, I want to incorporate the familiar frequentists' flexibility into the IM construction.  This allows the user to construct a \emph{generalized IM} without specifying a full data-generating mechanism, simplifying the construction in several ways.  First, just like in the likelihood-based approaches mentioned above, a generalized IM can be constructed based on the sampling model alone, or some function thereof, easing the burden on the user.  Second, the generalized IM can be constructed based on a \emph{generalized association} that involves only a one-dimensional auxiliary variable, which simplifies user's task of selecting a good predictive random set.  Compare this to the basic IM approach where the user must first specify a data-generating mechanism and carry out some potentially non-trivial dimension-reduction steps \citep[e.g.,][]{imcond}.  Despite making substantial simplifications to the IM construction, it can be shown that this generalization preserves the IM's guaranteed validity property under mild conditions.  Therefore, the generalized IM framework is a simple and widely applicable tool for valid, prior-free, probabilistic inference.  

This paper's main contribution is the new perspective it brings to some more-or-less familiar ideas, results, and techniques.  Specifically, all of the familiar considerations used in constructing statistical procedures fit within the the seemingly rigid IM framework, and this has at least two useful consequences.  First, working within the IM framework does not require that one abandon all the classical tools and ways of thinking---these can be merged seamlessly into the framework itself.  Second, new insights concerning these classical tools can be gained when looking from an IM point of view; see Section~\ref{SS:remarks}.  

The remainder of the paper is organized as follows.  After some background on IMs in Section~\ref{S:background}, the new generalized IM approach is presented in Section~\ref{S:gims}, with a motivating validity theorem and a special case that is relatively easy to implement, involving only a scalar auxiliary variable, and having good properties.  Some important practical considerations, namely, computation and marginalization, are discussed in Section~\ref{S:practical}, and two interesting and challenging applications---inference on the odds ratio in $2 \times 2$ tables and inference on the error variance in mixed-effects models---are presented in Section~\ref{S:examples}.  Concluding remarks are made in Section~\ref{S:discuss}.

\section{Background on IMs}
\label{S:background}

Let $Y \in \YY$ be the observable data, and write $\prob_{Y|\theta}$ for the sampling model, which depends on an unknown parameter $\theta \in \Theta$.  In the basic IM framework, described in \citet{imbasics}, the starting point---the \emph{A-step}---is to associate $Y$ and $\theta$ with an unobservable auxiliary variable $U \in \UU$ with known distribution $\prob_U$.  Formally, suppose the association can be written as
\begin{equation}
\label{eq:basic.assoc}
Y = a(\theta, U), \quad U \sim \prob_U. 
\end{equation}
\citet{imcond, immarg} argue that some dimension-reduction steps should be taken first before an association mapping is defined, so the left-hand side may be something different than the observable data, e.g., a minimal sufficient statistic.  This dimension-reduction step is recommended, but it is not necessary to describe these details here.  The result of the A-step is a set-valued mapping 
\begin{equation}
\label{eq:basic.focal}
\Theta_y(u) = \{\theta: y = a(\theta,u)\}, \quad u \in \UU, 
\end{equation}
indexed by the observed $Y=y$.  The main point is that the association determines the sampling model $\prob_{Y|\theta}$ or, alternatively, the ingredients in \eqref{eq:basic.assoc} must be chosen to be consistent with the given sampling model.  However, there may be several versions of the association that are consistent with the sampling model, and different versions may produce different inferences.  This is not unlike the frequentists' choice of (approximate) pivot for constructing a test, confidence region, etc.  In any case, the question of which association \eqref{eq:basic.assoc} to take, for given sampling model $\prob_{Y|\theta}$, is an important one.  

The second step in the basic IM construction---the \emph{P-step}---is to predict the unobserved value of $U$ in \eqref{eq:basic.assoc}, corresponding to the observed $Y=y$, with predictive random set $\S$.  The P-step is the defining feature of the IM framework, driving its essential properties and separating it from the approach described in \citet{dempster2008}.  The distribution $\prob_\S$ of $\S$ is to be chosen by the user, subject to a certain ``validity'' condition, namely, that, if $f_\S(u) = \prob_\S(\S \ni u)$, then 
\[ f_\S(U) \stgeq \unif(0,1), \quad \text{as a function of $U \sim \prob_U$}, \]
where ``$\stgeq \unif(0,1)$'' means ``stochastically no smaller than $\unif(0,1)$,'' i.e.,  
\begin{equation}
\label{eq:prs.valid}
\prob_U\bigl\{f_\S(U) \leq \alpha\bigr\} \leq \alpha, \quad \forall \; \alpha \in (0,1). 
\end{equation}
Intuitively, the random set $\S$ is meant to be ``good'' at predicting samples from $\prob_U$ and \eqref{eq:prs.valid} makes this precise: the $\prob_\S$-probability of the event ``$\S \ni u$'' is small only for a set of $u$ values with relatively small $\prob_U$-probability.  Sufficient conditions for \eqref{eq:prs.valid} are mild, so it is easy to find a valid predictive random set; in fact, most applications of IMs employ a simple ``default'' predictive random set, see \eqref{eq:default.prs}.  


The third and final step in the basic IM construction---the \emph{C-step}---is to combine the association at the observed data $Y=y$ with the predictive random set $\S$.  Specifically, one obtains a random subset of $\Theta$:
\begin{equation}
\label{eq:post.focal}
\Theta_y(\S) = \bigcup_{u \in \S} \Theta_y(u). 
\end{equation}
The intuition behind this is as follows: if one believes that $\S$ contains the value of $U$ corresponding to the observed $Y=y$ and the true $\theta$, which is justified by \eqref{eq:prs.valid}, then one must also believe, with equal conviction, that $\Theta_y(\S)$ contains the true $\theta$.  The IM output is the distribution of the random set $\Theta_y(\S)$, which I will summarize with a plausibility function.  Specifically, if $A \subset \Theta$, then the plausibility function at $A$ is 
\[ \pl_y(A) = \prob_\S\{\Theta_y(\S) \cap A \neq \varnothing\}. \]
Of course, the plausibility function depend on $\S$ or, more precisely, on $\prob_\S$, but I omit this dependence in the notation.  For interpretation, $\pl_y(A)$ is a measure of the degree of belief, given data $y$, in the falsity of ``$\theta \not\in A$.''  The user's ``belief'' is first encoded in $\prob_\S$, a personal or belief probability, subject to the constraint \eqref{eq:prs.valid}, which is then transferred to the parameter space in the IM's C-step.  Intuitively, it is possible that two disjoint assertions are highly plausible based on the given data, and the plausibility function allows for this, i.e., plausibility satisfies $\pl_y(A) + \pl_y(A^c) \geq 1$ for all $A$.  Moreover, Theorem~2 in \citet{imbasics} shows that if $\S$ satisfies \eqref{eq:prs.valid}, then $\pl_Y(A)$ is properly calibrated as a function of $Y \sim \prob_{Y|\theta}$ for fixed $A$, in the sense that 
\begin{equation}
\label{eq:valid}
\sup_{\theta \in A} \prob_{Y|\theta}\bigl\{ \pl_Y(A) \leq \alpha \bigr\} \leq \alpha, \quad \forall \; \alpha \in (0,1), \quad \forall \; A \subseteq \Theta, 
\end{equation}
or, in other words, for any $A \subseteq \Theta$, if $\theta \in A$, then $\pl_Y(A) \stgeq \unif(0,1)$, as a function of $Y \sim \prob_{Y|\theta}$.  When \eqref{eq:valid} holds, the IM is said to be \emph{valid}.  This validity property aids in interpreting the plausibility function values---it puts the personal/belief probabilities on an objective $\unif(0,1)$ scale---and also facilitates the construction of IM-based decision rules with guaranteed error rate control.

\ifthenelse{1=1}{}{
The IM output is a belief and plausibility function pair, basically the distribution of $\Theta_y(\S)$.  Specifically, if $A \subset \Theta$, then the belief and plausibility functions at $A$, respectively, are 
\[ \bel_y(A) = \prob_\S\{\Theta_y(\S) \subseteq A\} \quad \text{and} \quad \pl_y(A) = 1 - \bel_y(A^c). \]
Of course, the belief and plausibility functions depend on $\S$ or, more precisely, on $\prob_\S$, but I omit this dependence in the notation.  For interpretation, $\bel_y(A)$ is a measure of the user's degree of belief, given data $y$, in the truthfulness of the assertion ``$\theta \in A$,'' and $\pl_y(A)$ is a measure of the degree of belief, given data $y$, in the falsity of ``$\theta \not\in A$.''  The user's ``belief'' is first encoded in $\prob_\S$, a personal or belief probability, subject to the constraint \eqref{eq:prs.valid}, which is then transferred to the parameter space in the IM's C-step.  Intuitively, belief in an assertion need not correspond to belief against its complement, and it is easy to see that the belief and plausibility functions meet this intuition, i.e., $\bel_y(A) \leq \pl_y(A)$ for all $A$.  Moreover, Theorem~2 in \citet{imbasics} shows that if $\S$ satisfies \eqref{eq:prs.valid}, then $\bel_Y(A)$ and $\pl_Y(A)$ are properly calibrated, as functions of $Y \sim \prob_{Y|\theta}$ for fixed $A$, in the sense that 
\begin{equation}
\label{eq:valid}
\sup_{\theta \in A} \prob_{Y|\theta}\bigl\{ \pl_Y(A) \leq \alpha \bigr\} \leq \alpha, \quad \forall \; \alpha \in (0,1), \quad \forall \; A \subseteq \Theta, 
\end{equation}
or, in other words, for any $A \subseteq \Theta$, if $\theta \in A$, then $\pl_Y(A) \stgeq \unif(0,1)$, as a function of $Y \sim \prob_{Y|\theta}$.  When \eqref{eq:valid} holds, the IM is said to be \emph{valid}.  Of course, since it holds for all $A$, validity can also be defined in terms of $\bel_y$.  This validity property aids in interpreting the belief and plausibility function values---it puts the subjective/belief probabilities on an objective $\unif(0,1)$ scale---and also facilitates the construction of IM-based decision rules with guaranteed error rate control.  
}

The conclusion I hope the reader will reach from this brief summary is that the IM approach is conceptually straightforward and accomplishes what Fisher's fiducial approach was meant to, namely, valid prior-free probabilistic inference.  The apparent cost is that the IM output depends on the choice of association \eqref{eq:basic.assoc}, the choice of predictive random set, and, in a less-obvious way, on the dimension of the auxiliary variable.  The need to specify an association, carry out the necessary dimension-reduction steps, and introduce a valid predictive random set may give the impression that the IM approach is not user-friendly.  The goal of this paper is to show how one can construct a valid IM by dealing with these challenges indirectly.

\section{A class of generalized IMs}
\label{S:gims}

\subsection{Construction}
\label{SS:construction}

Towards accomplishing the goals laid out above, we discuss here how the basic association \eqref{eq:basic.assoc} can be made simpler and more flexible, by relaxing the direct connection with the sampling model and informally reducing auxiliary variable dimension, while still retaining the desirable validity properties of the resulting IM.  

Start by going back to the beginning of Section~\ref{S:background} where only the sampling model $\prob_{Y|\theta}$ for data $Y$ given parameter $\theta$ is available.  The IM construction in Section~\ref{S:background} is based on identification of an unobservable auxiliary variable $U$ to associate with $(Y,\theta)$ and then to be predicted.  The basic approach identifies $U$ by thinking about the data-generating process, but this is potentially restrictive and unnecessary.  Rather than specifying a potentially relatively high-dimensional auxiliary variable corresponding to a data-generation process, and then subsequently reducing the dimension according to guidelines in \citet{imcond, immarg}, is it possible to specify an auxiliary variable of the appropriate dimension directly and easily?

Towards answering this question, the key insight is that the association in \eqref{eq:basic.assoc} need not involve the full data $Y$.  For a function $(y,\theta) \mapsto T_{y,\theta}$, consider a \emph{generalized association}
\begin{equation}
\label{eq:ga}
T_{Y,\theta} = a(\theta, U), \quad U \sim \prob_U, 
\end{equation}
where $U$ is some auxiliary variable taking values in a space $\UU$.  Note that, unless $y \mapsto T_{y,\theta}$ is one-to-one for each $\theta$, which is not a useful case, the generalized association does not determine the sampling model for $Y$, thereby relaxing the requirement in Section~\ref{S:background} that the association specify a version of the data-generating mechanism.  It does, however, determine the sampling model of $T_{Y,\theta}$ under $Y \sim \prob_{Y|\theta}$, so \eqref{eq:ga} is compatible with $\prob_{Y|\theta}$ in this sense.  The function $T_{Y,\theta}$ can depend on $\theta$ or not, and its distribution need not be continuous.  Some examples are discussed below and in the later sections.  

Based on \eqref{eq:ga}, the (generalized) A-step defines the set-valued mapping
\begin{equation}
\label{eq:gen.focal}
\Theta_y(u) = \{\theta: T_{y,\theta} = a(\theta,u)\}, \quad (y,u) \in \YY \times \UU. 
\end{equation}
Then the P- and C-steps can be carried out exactly like in Section~\ref{S:background}.  In particular, the P-step introduces a valid random set $\S \sim \prob_\S$ for predicting the unobserved value of $U$ in \eqref{eq:ga}, and the C-step yields the random set $\Theta_y(\S)$ as in \eqref{eq:post.focal} and the corresponding belief and plausibility functions $\bel_y$ and $\pl_y$, depending implicitly on $\prob_\S$.  I will call the resulting IM a \emph{generalized IM} and, interestingly, validity of this generalized IM, in the sense of \eqref{eq:valid}, follows immediately from the construction.  

\begin{thm}
\label{thm:valid}
For the generalized association \eqref{eq:ga}, let $\S \sim \prob_\S$ be a valid predictive random set for $U \sim \prob_U$.  If $\Theta_y(\S) \neq \varnothing$ with $\prob_\S$-probability~1 for all $y$, then the generalized IM is valid in the sense of \eqref{eq:valid}.  
\end{thm}

\begin{proof}
For any $A$, let $(y,\theta,u)$ be such that $\theta \in A$ and $T_{y,\theta} = a(\theta,u)$.  Since $\{\theta\} \subset A$ and $\pl_y$ is monotone, we have 
\[ \pl_y(A) \geq \pl_y(\{\theta\}) = \prob_\S\{\Theta_y(\S) \ni \theta\} = \prob_\S\{\S \ni u\} = f_\S(u). \]
Since $f_\S(U) \stgeq \unif(0,1)$ as a function of $U \sim \prob_U$, it follows that $\pl_Y(A) \stgeq \unif(0,1)$ as a function of $Y \sim \prob_{Y|\theta}$, i.e., $\prob_{Y|\theta}\{\pl_Y(A) \leq \alpha\} \leq \alpha$, for all $\alpha \in (0,1)$.  This holds for all $\theta \in A$, so take supremum of the left-hand side over $\theta \in A$ to complete the proof.  
\end{proof}

Therefore, construction of a valid generalized IM is possible and seems to be fairly straightforward.  An important consequence of the validity theorem is that plausibility regions based on the generalized IM have the nominal coverage probability.  That is, if 
\[ \C_\alpha(y) = \{\theta: \pl_y(\theta) > \alpha\}, \]
where $\pl_y(\theta) = \pl_y(\{\theta\})$, then  
\[ \prob_{Y|\theta}\{\C_\alpha(Y) \ni \theta\} = \prob_{Y|\theta}\{\pl_Y(\theta) > \alpha\}, \]
and since the validity property \eqref{eq:valid} holds for all $A$, in particular, $A=\{\theta\}$, we get that the right-hand side in the above display is $\geq 1-\alpha$ for all $\theta$.  An important observation is that this does not require large samples or any assumptions on the model.

\subsection{A useful special case}
\label{SS:special}

There are, of course, a variety of ways one can specify the generalized association \eqref{eq:ga}.  Here I will elaborate on one simple but general strategy.  Let $(y,\theta) \mapsto T_{y,\theta}$ be scalar-valued, e.g., the likelihood ratio or a function thereof; in general, it is not a statistic because it depends on $\theta$.  Moreover, since the map is scalar-valued, in most cases, it cannot be one-to-one so it corresponds to a non-trivial summary of the data $y$.  Suppose that $T_{Y,\theta}$ has a continuous distribution, under $Y \sim \prob_{Y|\theta}$, and let $F_\theta$ be the corresponding distribution function.  Now specify an association in terms of the distribution of $T_{Y,\theta}$: 
\begin{equation}
\label{eq:Ty.assoc}
T_{Y,\theta} = F_\theta^{-1}(U), \quad U \sim \unif(0,1). 
\end{equation}
The case of discrete $T_{Y,\theta}$ can be handled similarly, i.e., 
\[ F_\theta(T_{Y,\theta}-) \leq U < F_\theta(T_{Y,\theta}), \quad U \sim \unif(0,1), \]
where $F_\theta(t-) = \lim_{s \uparrow t} F_\theta(s)$ is the left-hand limit.  This corresponds to taking $a(\theta,u)$ in \eqref{eq:ga} to be $F_\theta^{-1}(u)$.  Now, with a suitable predictive random set for $U \sim \unif(0,1)$, this generalized association leads to a valid generalized IM.  

\begin{cor}
The generalized IM constructed based on the association \eqref{eq:Ty.assoc} and a valid predictive random set $\S$ for $U \sim \unif(0,1)$ is valid in the sense of Theorem~\ref{thm:valid}, provided that $\Theta_y(\S) \neq \varnothing$ with $\prob_\S$-probability~1 for all $y$.  
\end{cor}

This provides a simple and general procedure for constructing a valid generalized IM based on a choice of mapping $T_{y,\theta}$.  In fact, this shows that the work done by \citet{plausfn} in the frequentist context is just a special case of the proposed generalized IM framework.  His choice to work primarily with the negative log-likelihood ratio, 
\begin{equation}
\label{eq:lrt}
T_{y,\theta} = -2\log \frac{L_y(\theta)}{\sup_{\vartheta \in \Theta} L_y(\vartheta)}, 
\end{equation}
with $L_y$ the likelihood function for $\theta$ based on data $Y=y$.  The $T_{y,\theta}$ in\eqref{eq:lrt} is the deviance used frequently in \citet{schweder.hjort.book}.  Also, other authors, e.g., \citet{wasserman1990b}, \citet{aickin2000}, and \citet{denoeux2014}, have used the likelihood ratio to construct a plausibility function for statistical inference, but in a different way than I propose here.  There are, however, other choices of $T_{y,\theta}$; see, e.g., Remark~\ref{re:efficiency} and Section~\ref{S:examples}. 

A natural question is if anything is gained from the generalized IM perspective, besides the apparent simplicity, compared to the basic IM approach described in Section~\ref{S:background} and the references therein.  The next example demonstrates that the simple generalized IM can lead to improved efficiency, at least in some cases.  

\begin{ex}
\label{ex:gamma}
Let $Y_1,\ldots,Y_n$ be iid samples from a $\gam(\theta_1,\theta_2)$ distribution, where $\theta_1$ is the shape parameter and $\theta_2$ is the scale parameter, both unknown.  This same problem was considered \citet[][Section~5.3]{imcond} and they presented a basic IM solution based on a reduction to the complete sufficient statistic.  This requires specifying a predictive random set for a two-dimensional auxiliary variable consisting of two independent uniforms.  No IM optimality results are available for this problem, so they made the natural choice of a square-shaped predictive random set.  This guarantees validity of the IM, but efficiency is a question.  For comparison, consider a generalized IM based on the likelihood ratio, which is also valid; the computational details are discussed in Section~\ref{SS:computation}.  I simulate $n=25$ observations from the gamma distribution with $\theta_1=7$ and $\theta_2=3$.  Figure~\ref{fig:gamma} displays several results: the Jeffreys prior Bayesian posterior samples, the 90\% confidence ellipse based on asymptotic normality of the maximum likelihood estimator, the 90\% confidence region based on the asymptotic chi-square distribution of the deviance \citep{schweder.hjort.book}, the 90\% plausibility region based on the IM construction in \citet{imcond}, and the 90\% plausibility region based on the likelihood ratio-based generalized IM.  Interestingly, the generalized IM plausibility region has guaranteed 90\% coverage and it captures the overall shape of the posterior, which is non-elliptical.  It is also slightly smaller than the deviance-based region and is considerably smaller than the basic IM plausibility region.  
Another exact confidence region for this gamma problem is obtained in \citet{taraldsen.lindqvist.2013}.  
\end{ex}

\begin{figure}
\begin{center}
\scalebox{0.6}{\includegraphics{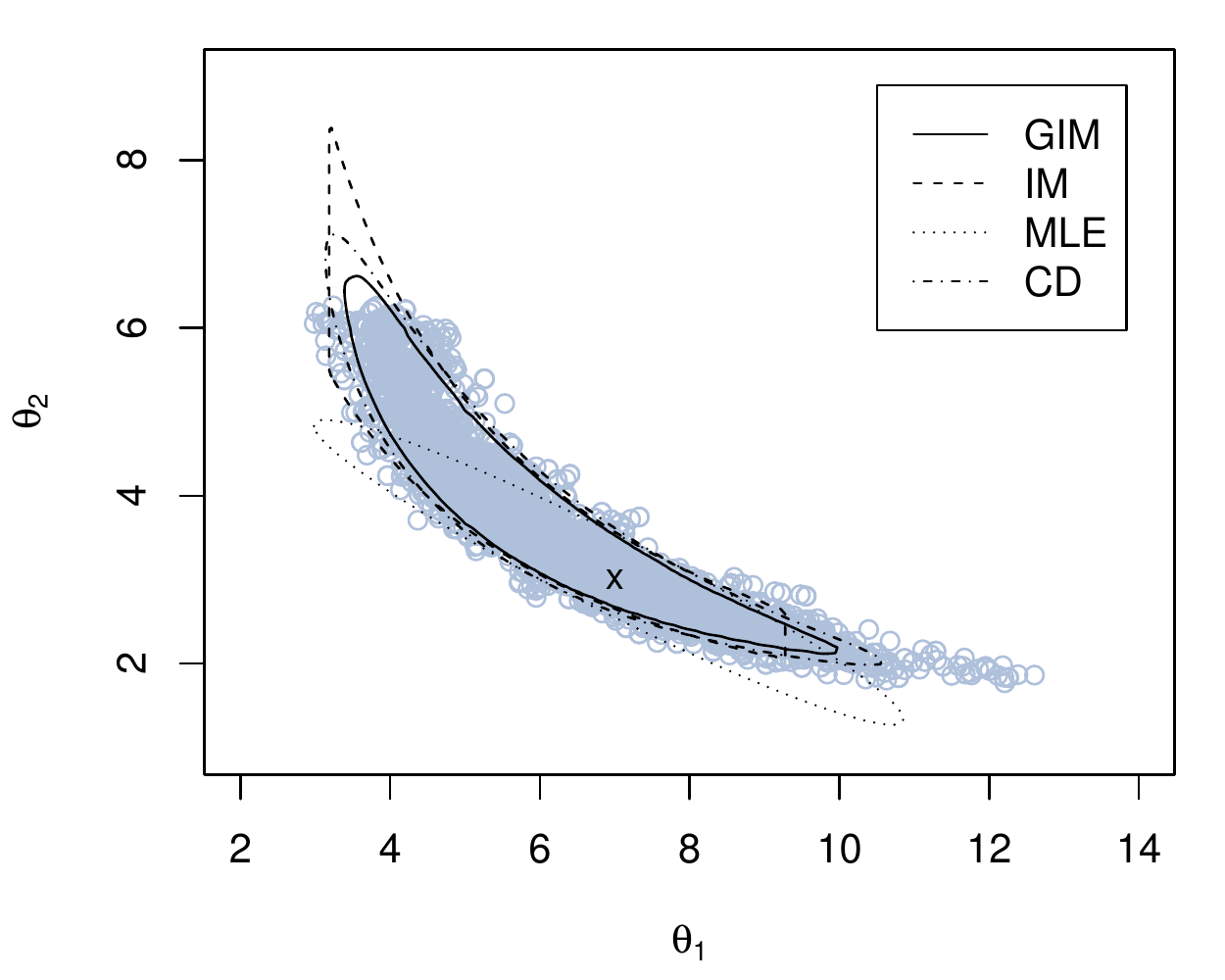}}
\end{center}
\caption{Output from the gamma simulation in Example~\ref{ex:gamma}: Jeffreys prior Bayes posterior samples (gray); maximum likelihood-based 90\% confidence ellipse (dotted); 90\% deviance-based confidence region (dot-dashed); and 90\% plausibility regions based on basic (dashed) and generalized (solid) IMs.}
\label{fig:gamma}
\end{figure}

\subsection{Remarks}
\label{SS:remarks}

\begin{remark}[on asymptotics]
\label{re:asymptotics}
To make this discussion concrete, consider the case where $Y$ consists of a collection of $n$ iid observations.  When $n$ is large, there is no shortage of pivotal quantities $T_{Y,\theta}$ that can be used in the generalized association \eqref{eq:Ty.assoc}.  Indeed, Wilks's theorem says that $T_{Y,\theta}$ in \eqref{eq:lrt} has an asymptotic chi-square distribution under $\prob_{Y|\theta}$, as $n \to \infty$.  In this case the $F_\theta$ in \eqref{eq:Ty.assoc} can, asymptotically, be taken as a suitable chi-square distribution function, free of $\theta$.  The same holds in the case with nuisance parameters using a profile likelihood, as in \eqref{eq:profile}.  There are many other choices of $T_{Y,\theta}$ that are asymptotic pivots, e.g., the quantities in \citet[][Chap.~8]{brazzale.davison.reid.2007} with higher-order approximation accuracy.  The point is that the generalized IM framework provides a tool for valid statistical inference without appealing to asymptotics but, if desired, asymptotic theory can be used just to provide simple large-sample approximations.
\end{remark}

\begin{remark}[on confidence distributions]
\label{re:confidence}
Confidence distributions \citep{xie.singh.2012, schweder.hjort.2002, schweder.hjort.book, singh.xie.strawderman.2007} have received considerable attention recently, especially in the meta-analysis context \citep{liu.liu.xie.2014, yang.liu.liu.xie.2014, claggett.xie.tian.2014, liu.liu.xie.2015, xie.singh.strawderman.2011}, a primary selling point being that it ``unifies'' \citep[][p.~3]{xie.singh.2012} existing approaches.  Their point is that a variety of standard tools can be converted into a confidence distribution or an asymptotic confidence distribution.  My proposal here for a generalized IM can be interpreted similarly, since many familiar ideas from classical statistics can be employed to construct a valid generalized IM.  
\end{remark}

\begin{remark}[on efficiency and choice of $T_{Y,\theta}$]
\label{re:efficiency}
Towards an optimal IM, \citet{imbasics} suggested that, for a fixed $\theta_0$, the best random set $\S$ is one that makes $\pl_Y(\theta_0)$ as stochastically as small as possible, subject to the validity condition.  They argue that there exists a nested collection $\YY_\alpha \subset \YY$, depending on $\S$ and $\theta_0$, such that $\pl_y(\theta_0) > \alpha$ if and only if $y \in \YY_\alpha$ and, furthermore, the optimal $\S$ has corresponding $\YY_\alpha$ such that 
\[ \int_{\YY_\alpha} S_\theta(y) \, p_\theta(y) \,dy = 0 \quad \text{at $\theta=\theta_0$ for all $\alpha$}, \]
where $p_\theta$ is the density function for $Y$ and $S_\theta(y) = (\partial/\partial\theta) \log p_\theta(y)$ is the familiar score function.  Since $\E_{\theta_0}\{S_{\theta_0}(Y)\} = 0$, this condition implies that $\YY_\alpha$ is suitably balanced with respect to the distribution of $S_{\theta_0}(Y)$; this is called a \emph{score-balance condition}.  A set that will satisfy the score-balance condition, at least asymptotically, is 
\[ \YY_\alpha = \{y: S_{\theta_0}(y)^\top I(\theta_0)^{-1} S_{\theta_0}(y) \leq c_\alpha\} \]
for suitable constant $c_\alpha$, where $I(\theta)$ is the Fisher information.  This suggests choosing 
\[ T_{Y,\theta_0} = S_{\theta_0}(Y)^\top I(\theta_0)^{-1} S_{\theta_0}(Y), \]
and the corresponding plausibility function matches (asymptotically) the p-value of Rao's score test, which has certain optimality properties.  This provides some insight into the choice of an efficient mapping $T_{y,\theta}$, but more work is needed.  How the optimality considerations in the confidence distribution context \citep[e.g.,][Ch.~5]{schweder.hjort.book} might be useful in the IM context deserves further investigation.  
\end{remark}

\section{Practical considerations}
\label{S:practical}

\subsection{Computation}
\label{SS:computation}

For the case \eqref{eq:Ty.assoc}, suppose that large values of $T_{y,\theta}$ are suggestive that the model $\prob_{Y|\theta}$ does not fit data $Y=y$ well.  The log-likelihood ratio in \eqref{eq:lrt}, the score-balanced cased in Remark~\ref{re:efficiency}, among others, are of this form.  In this case, a natural choice of the random set $\S$ is the one-sided (nested) random interval
\[ \S = [0,U], \quad U \sim \unif(0,1). \]
With this choice, 
\begin{align*}
\Theta_y(\S) \cap A \neq \varnothing & \iff \{\theta: F_\theta(T_{y,\theta}) \leq U\} \cap A \neq \varnothing \\
& \iff \{U \geq F_\theta(T_{y,\theta}), \; \exists \; \theta \in A\} \\
& \iff \Bigl\{U \geq \inf_{\theta \in A} F_\theta(T_{y,\theta}) \Bigr\}
\end{align*}
and, therefore, the corresponding plausibility function is 
\begin{equation}
\label{eq:pl}
\pl_y(A) = \prob_\S\{\Theta_y(\S) \cap A \neq \varnothing \} = 1 - \inf_{\theta \in A} F_\theta(T_{y,\theta}) = \sup_{\theta \in A} \bar F_\theta(T_{y,\theta}), 
\end{equation}
where $\bar F_\theta = 1-F_\theta$ is the survival function.  Of course, for singleton assertions, no optimization is necessary.  The point is that evaluating the generalized IM plausibility function requires only some relatively simple probability calculations.  

In cases where the distribution function $F_\theta$ is not available in closed form, a conceptually simple Monte Carlo approximation is available:
\begin{equation}
\label{eq:monte1}
F_\theta(t) \approx \frac1M \sum_{m=1}^M I\{T_{Y^{(m)},\theta} \leq t\}, 
\end{equation}
where $\{Y^{(m)}: m=1,\ldots,M\}$ are independent copies of $Y \sim \prob_{Y|\theta}$.  Of course, if direct information about the distribution of $T_{Y,\theta}$ is available, e.g., that it depends only on some function of $Y$, then this can be used to avoid simulation of the entire $Y$.  This approach is straightforward, but can be time-consuming to implement because the plausibility function may need to be evaluated at many different $\theta$ values, and each requires its own Monte Carlo simulation.  This difficulty can be avoided if it were possible to simulate from $\prob_{Y|\theta}$ for only a single value of $\theta$.  One way this can be achieved is if it happens that the distribution of $T_{Y,\theta}$, under $Y \sim \prob_{Y|\theta}$, does not depend on $\theta$, i.e., $F_\theta \equiv F$.  This invariance property holds if $T_{Y,\theta}$ is itself a pivot, which can be arranged in some examples \citep[e.g.,][Sec.~2.4]{plausfn}.  More generally, an importance sampling strategy can be employed to approximate $F_\theta$ over a range of $\theta$ with only a single Monte Carlo sample.  Choose a fixed parameter value, say, $\hat\theta$, a suitable estimator, and rewrite \eqref{eq:monte1} as 
\[ F_\theta(t) \approx \frac1M \sum_{m=1}^M I\{T_{Y^{(m)},\theta} \leq t\} \frac{L_{Y^{(m)}}(\theta)}{L_{Y^{(m)}}(\hat\theta)}, \]
where, this time, $\{Y^{(m)}: m=1,\ldots,M\}$ are independent samples from $\prob_{Y|\hat\theta}$, and can be reused for different values of $\theta$.  This is reminiscent of parametric bootstrap \citep[e.g.,][]{davison.hinkley.1997}, and will have a much smaller computational cost compared to the naive Monte Carlo approximation in \eqref{eq:monte1}.  The two Monte Carlo strategies discussed here are extreme in the sense that the former takes a Monte Carlo sample for each $\theta$ while the latter takes only one Monte Carlo sample for a single $\hat\theta$.  Various middle ground strategies are also possible, e.g., take Monte Carlo samples for a fixed grid $\{\vartheta_1,\ldots,\vartheta_G\}$ of parameter values and do an importance sampling-based approximation of $F_\theta(t)$ using samples corresponding to grid point $\vartheta_g$ where $g=\arg\min_h\|\theta-\vartheta_h\|$.  

\begin{ex}
\label{ex:triangle}
An interesting non-standard example is the so-called asymmetric triangular distribution \citep[e.g.,][Example~11]{bergerbernardosun2009}, with density function
\[ p_\theta(y) = \begin{cases} 2y/\theta  & \text{if $0 \leq y \leq \theta$}, \\ 2(1-y) / (1-\theta)  & \text{if $\theta < y \leq 1$}, \end{cases} \]
where $\theta \in [0,1]$.  The density has a unique mode at $\theta$, but the density has a corner and is not differentiable there.  Consider making inference on $\theta$ based on an independent sample $Y=(Y_1,\ldots,Y_n)$.  This is a challenging problem because there is no non-trivial sufficient statistic and the formal Fisher information is not well-defined.  Constructing an efficient IM using the basic approach outlined in Section~\ref{S:background} is difficult because there is no clear strategy to reduce the dimension of the auxiliary variable.  However, a generalized IM for $\theta$ is readily available here using the likelihood ratio \eqref{eq:lrt} as in Section~\ref{SS:special}.  For a quick comparison of the generalized IM (based on importance sampling) and the confidence distribution based on the asymptotic chi-square distribution of the deviance in \citet{schweder.hjort.book}, data $Y$ of size $n=10$ is simulated from the triangular distribution with $\theta=0.3$.  Plots of the plausibility and confidence curves are shown in Figure~\ref{fig:triangle}.  The two curves have roughly the same shape, though the confidence curve is a bit tighter, a consequence of the overly optimistic asymptotic approximation.     
\end{ex}

\begin{figure}
\begin{center}
\scalebox{0.6}{\includegraphics{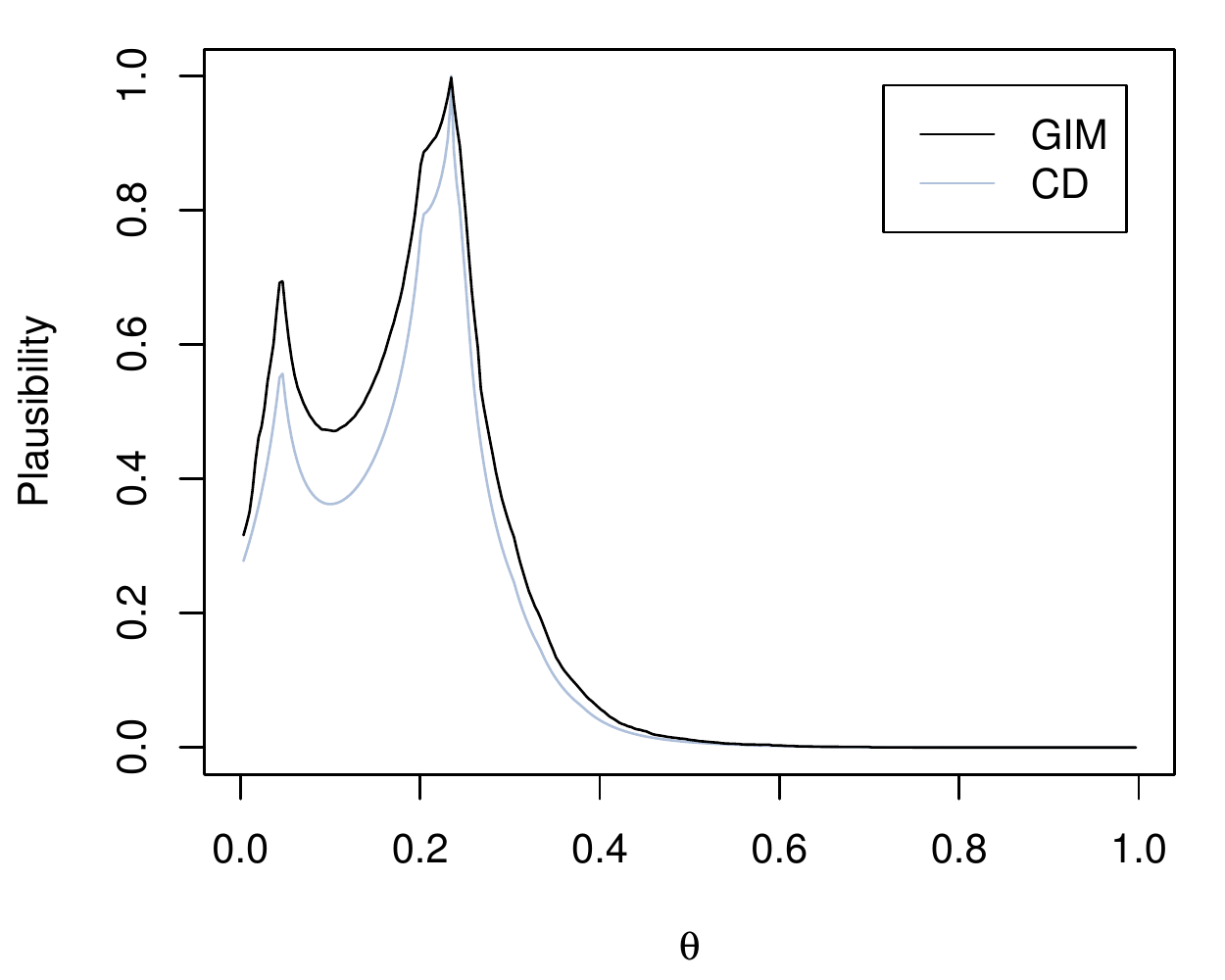}}
\end{center}
\caption{Plots of the plausibility function $\pl_y(\theta)$ for the triangular model in Example~\ref{ex:triangle} based on naive (dashed) and importance sampling-driven (solid) Monte Carlo.}
\label{fig:triangle}
\end{figure}

Though the context here is a bit different, the use of Monte Carlo methods to construct tests and confidence regions has been addressed previously in the literature.  For example, our $100(1-\alpha)$\% plausibility regions correspond to finding solutions to the equation $\pl_y(\theta) = \alpha$.  When the plausibility function can only be evaluated via Monte Carlo, solving this equation is a stochastic approximation problem \citep{robbinsmonro}, and has been discussed in \citet{garthwaite.buckland.1992} and \citet{botev.lloyd.2015}; see, also, \citet{bolviken.skovlund.1996}.  

Another issue to address is optimization of the function $F_\theta(T_{y,\theta})$ over a subset $A$ of $\theta$ values.  This will again be relevant in the discussion of marginalization below.  Recently, but again in a slightly different context, \citet{xiong.optim} considers this optimization problem and suggests some localization strategies as well as a proper choice of grid points, via space-filling designs, on which the plausibility function surface can be built up.

\subsection{Handling nuisance parameters}
\label{SS:nuisance}

Most practical problems involve nuisance parameters, so having some general techniques to eliminate these parameters is important.  Without loss of generality, partition the full parameter $\theta$ as $\theta=(\psi, \lambda)$, where $\psi$ is the interest parameter and $\lambda$ is the nuisance parameter.  Here I will discuss three different approaches for eliminating $\lambda$ in to construct a marginal generalized IM for $\psi$.  

A first strategy is conditioning.  In particular, let $y \mapsto (T_y, T_y')$ be a one-to-one transformation of $y$, independent of the parameter $\theta$.   If the conditional distribution of $T_y$, given $T_y'$, is free of $\lambda$, then this conditional distribution can be used to construct a generalized IM for $\psi$.  Section~\ref{SS:odds} presents an example of this conditioning strategy in action.   

The second strategy is a direct marginalization by selecting a function $T_{Y,\psi}$, depending on $Y$ and $\psi$ only, such that its distribution is free of the nuisance parameter $\lambda$.  A general candidate for such a function, generalizing the idea at the end of Section~\ref{S:gims}, is the profile likelihood ratio
\begin{equation}
\label{eq:profile}
T_{Y,\psi} = -2\log \frac{\sup_\lambda L_y(\psi, \lambda)}{\sup_{\psi, \lambda} L_y(\psi, \lambda)}. 
\end{equation}
Composite transformation models \citep{bn1988} form a general class of problems where this approach to marginalization can be applied.  For example, in the two-parameter gamma model, where $\psi$ is the shape, the Bartlett test statistic has distribution free of the nuisance scale parameter.  Similarly, in the bivariate normal model, where $\psi$ is the correlation, the sample correlation coefficient $\hat \psi$ has distribution free of $\lambda$, the means and variances; the profile likelihood is a function of only $\hat\psi$ and $\psi$ and, therefore, also has distribution free of $\lambda$.  A mixed-effects model, where the nuisance fixed-effect parameters are eliminated via marginalization, is presented in Section~\ref{SS:vc}.

A third strategy, which seems to be unique to the framework presented here, is a different form of marginalization via optimization.  When the underlying random sets are nested, which is the recommended choice, the plausibility function is called \emph{consonant} \citep[e.g.,][]{shafer1987}.  In particular, this means that the plausibility function evaluated at a set $A$ equals the suprema of the plausibility function evaluated at points in $A$.  This provides some further explanation for the expression for $\pl_y(A)$ in \eqref{eq:pl} involving a supremum.  This is relevant in the present situation because a problem that involves nuisance parameters can be handled by considering assertions about the full parameter $(\psi,\lambda)$ that span the full range of $\lambda$.  Therefore, marginalization can be accomplished by optimization \emph{after} evaluating the plausibility function, compared to the pre-plausibility evaluation optimization in the profiling approach discussed above.  This further demonstrates the importance of the optimization aspects discussed in Section~\ref{SS:computation}.  It is preferable to eliminate the nuisance parameters before evaluating plausibility, if possible, because it reduces the computational cost, but for some problems there are no obvious conditioning or profiling strategies to use, so this default marginalization tool is necessary.  

\section{Applications}
\label{S:examples}

\subsection{Odds ratio in a $2 \times 2$ table}
\label{SS:odds}

Let $Y=(Y_0,Y_1)$ be two independent binomial counts, with $Y_0 \sim \bin(n_0, \theta_0)$ and $Y_1 \sim \bin(n_1,\theta_1)$, where $n=(n_0,n_1)$ is known but $\theta=(\theta_0,\theta_1)$ is unknown.  Data such as these arise in, say, a clinical trial, where $Y_0$ and $Y_1$ correspond to the number of events observed under the control and treatment.  Suppose that the quantity of interest is the odds ratio
\[ \psi = \frac{\theta_1 / (1-\theta_1)}{\theta_0 / (1-\theta_0)}. \]
As in \citet{hannig.xie.2012}, a key observation is that the conditional distribution of $Y_1$, given $Y_0+Y_1$, depends on $\psi$ only, not on the nuisance parameter $\theta_0$ (or $\theta_1$), though the distribution form is not a standard one.  In particular, 
\[ \prob(Y_1 = y_1 \mid Y_0 + Y_1 = t) \propto \binom{n_1}{y_1} \binom{n_0}{t - y_1} \psi^{y_1},  \]
with $y_1$ ranging over $\max\{t-n_0, 0\}$ and $\min\{n_1,t\}$.  As discussed in Section~\ref{SS:nuisance}, let $T_Y = Y_1$ and $T_Y' = Y_0 + Y_1$.  For the observed value $t$ of $T_Y'$, let $F_{t,\psi}$ be the conditional distribution function corresponding to the mass function in the above display.  The resulting generalized association is 
\[ F_{t,\psi}(Y_1-1) \leq U < F_{t,\psi}(Y_1), \quad U \sim \unif(0,1), \]
and the A-step yields the sets 
\[ \Psi_y(u) = \{\psi: F_{t,\psi}(y_1-1) \leq u < F_{t,\psi}(y_1)\}, \quad u \in (0,1). \]
For predicting the value of this uniform auxiliary variable, a reasonable choice of predictive random set is the ``default'' \citep{imbasics} 
\begin{equation}
\label{eq:default.prs}
\S = \bigl[ 0.5 - |U - 0.5|, \, 0.5 + |U - 0.5| \bigr], \quad U \sim \unif(0,1). 
\end{equation}
Then the C-step combines $\Psi_y(\cdot)$ and $\S$ to get $\Psi_y(\S) = \bigcup_{u \in \S} \Psi_y(u)$.  Since 
\[ \Psi_y(\S) \not\ni \psi \iff F_{t,\psi}(y_1-1) > \sup\S \text{ or } F_{t,\psi}(y_1) < \inf \S, \]
we find that the corresponding plausibility function for singleton $\psi$ is 
\begin{align*}
\pl_y(\psi) & = \prob_\S\{\Psi_y(\S) \ni \psi \} \\
& = 1-\prob_U\{0.5 + |U-0.5| < F_{t,\psi}(y_1-1) \text{ or } 0.5 - |U-0.5| > F_{t,\psi}(y_1)\} \\
& = 1 - \prob_U\{|2U-1| < 2 F_{t,\psi}(y_1 - 1) - 1\} - \prob_U\{|2U-1| < 1 - F_{t,\psi}(y_1)\} \\ 
& = 1 - \{2 F_{t,\psi}(y_1 - 1) - 1\}^+ - \{1 - 2F_{t,\psi}(y_1)\}^+,  
\end{align*}
where the ``+'' superscript denotes the positive part.  The somewhat unusual form of this plausibility function is a result of the discreteness of the conditional distribution.  Some similar conditioning arguments are used in \citet{jin.li.jin.2015} to construct an IM for a different version of this discrete problem.  

For illustration, I consider two mortality data sets presented in Table~1 of \citet{normand1999}, namely, Trials~1 and 6.  Plausibility function for $\log \psi$ for the two data sets are displayed in Figure~\ref{fig:or}.  Both data sets have a relatively small numbers of events, and the two estimated odds ratios are similar: 2.27 in Trial~1 and 2.73 in Trial~6.  However, Trial~6 is an overall larger study, so the plausibility function is much more concentrated than that for Trial~1.  The flat peak is a result of the discreteness of the problem.  These plausibility function plots look quite different than those in Figure~2 \citet{hannig.xie.2012}, based on p-values from Fisher's exact test, in part because they make a certain correction to try to cancel out the effect of the discreteness.   

\begin{figure}
\begin{center}
\subfigure[Trial~1: $Y=(1,2)$, $n=(43, 39)$]{\scalebox{0.6}{\includegraphics{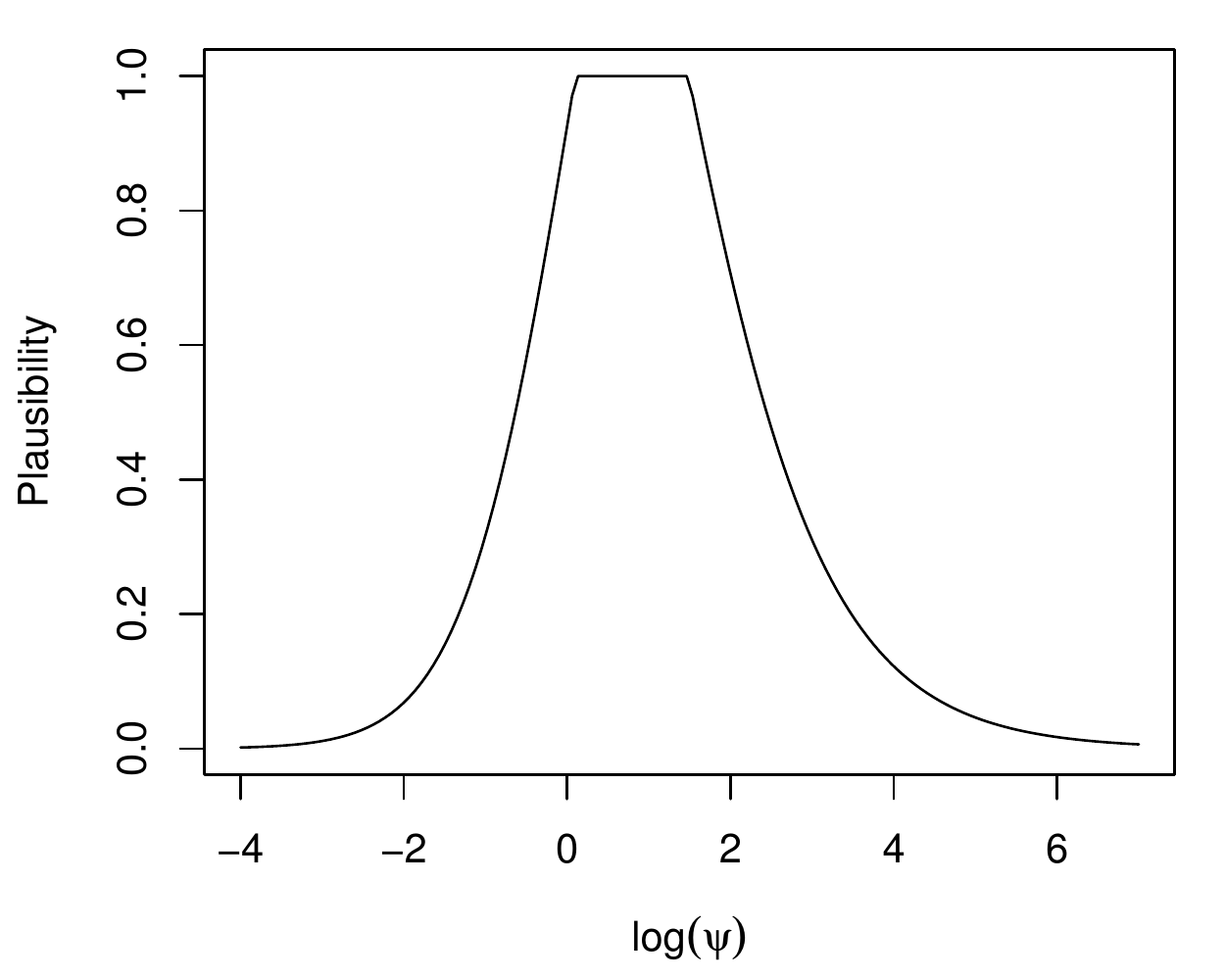}}}
\subfigure[Trial~6: $Y=(4,11)$, $n=(146, 154)$]{\scalebox{0.6}{\includegraphics{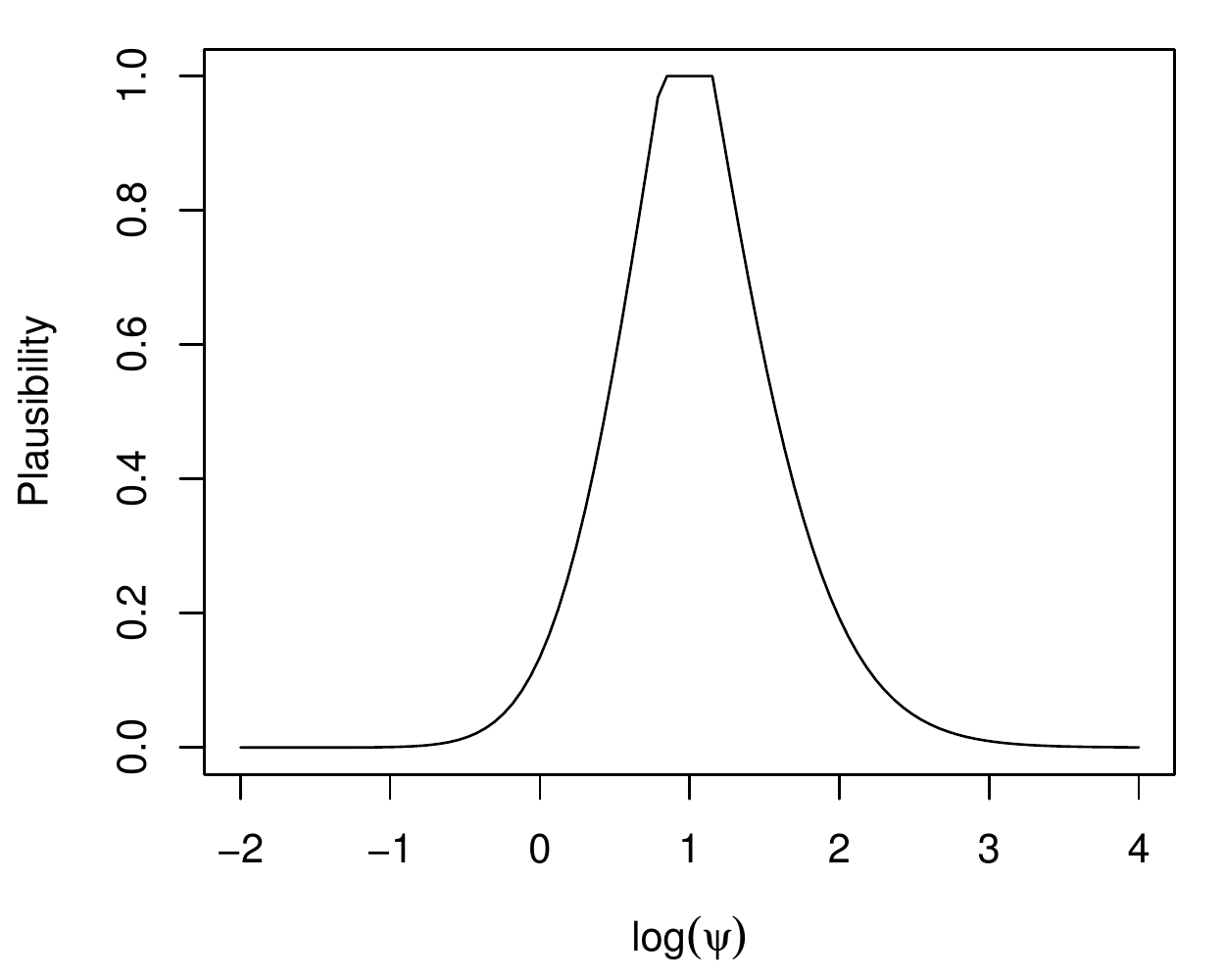}}}
\end{center}
\caption{Plausibility function for the log odds ratio in two mortality data sets (Trial~1 and Trial~6) presented in Table~1 of \citet{normand1999}.}
\label{fig:or}
\end{figure}


\subsection{Error variance in a mixed-effects model}
\label{SS:vc}

Consider a (possibly unbalanced) normal linear mixed effect model with two variance components, as in \citet{burch.iyer.1997}.  The model is written as 
\[ Y \sim \nm_n(X\beta, \sigma_\eps^2 I_n + \sigma_\alpha^2 Z A Z^\top), \]
where $X$ and $Z$ are $n \times p$ and $n \times a$ matrices of predictor variables, respectively, and the parameter is $\theta = (\beta, \sigma_\alpha^2, \sigma_\eps^2)$, with $\beta$ being the fixed-effect regression coefficients and $\avar$ and $\evar$ the random-effects and error variances, respectively.  Suppose that $\psi=\sigma_\eps^2$ is the parameter of interest and $\lambda=(\beta, \sigma_\alpha^2/\sigma_\eps^2)$ is a nuisance parameter.  Inference on $\psi$ is interesting from a theoretical point of view because, to my knowledge, there is no method available that can do this exactly; see, also, \citet[][p.~855]{e.hannig.iyer.2008}.  In what follows, I also assume that $X$ has full rank $p < n$ and that the matrix $A$, which describes the correlation structure in the random effects, is known.  

Following the setup in \citet{e.hannig.iyer.2008}, let $K$ be a $n \times (n-p)$ matrix such that $KK^\top = I_n - X(X^\top X)^{-1}X^\top$ and $K^\top K = I_{n-p}$.  It follows that 
\[ K^\top Y \sim \nm_{n-p}(0, \evar I_{n-p} + \avar G), \]
where $G = K^\top Z A Z^\top K$ is $(n-p) \times (n-p)$.  Let $e_1 > \cdots > e_L \geq 0$ denote the (distinct) eigenvalues of $G$ with multiplicities $r_1,\ldots,r_L$, respectively.  Let $P=[P_1,\ldots,P_L]$ be a $(n-p) \times (n-p)$ orthogonal matrix such that $P^\top G P$ is diagonal with eigenvalues $e_1,\ldots,e_L$, in their multiplicities, on the diagonal.  For $P_\ell$, a $(n-p) \times r_\ell$ matrix, define 
\[ S_\ell = Y^\top K P_\ell P_\ell^\top K^\top Y, \quad \ell=1,\ldots,L. \]
\citet{olsen.seely.birkes.1976} showed that $(S_1,\ldots,S_L)$ is a minimal sufficient statistic for $(\avar, \evar)$ and, moreover, its distribution is characterized by the equations 
\[ S_\ell = (\sigma_\alpha^2 e_\ell + \sigma_\eps^2) V_\ell, \quad \ell=1,\ldots,L, \]
where $V_1,\ldots,V_L$ are independent with $V_\ell \sim \chisq(r_\ell)$.  By making the transformation from $Y$ to $(S_1,\ldots,S_L)$, the nuisance fixed-effect parameter has been eliminated, as discussed in Section~\ref{SS:nuisance}.  With a slight abuse of notation, let $\lambda=\sigma_\alpha^2/\sigma_\eps^2$ be the remaining nuisance parameter.  Then the above equation can be rewritten as 
\[ S_\ell = \psi(\lambda e_\ell + 1) V_\ell, \quad \ell=1,\ldots,L. \]
In what follows, I propose a generalized IM for $\psi$ using some specialized tricks to eliminate the dependence on $\lambda$ as much as possible before full marginalization via optimization as discussed in Section~\ref{SS:nuisance}.  

Let $\L$ be a proper subset of $\{1,2,\ldots,L\}$, and write $H(\cdot \mid \lambda) = H_\L(\cdot \mid \lambda)$ for the distribution function of $\sum_{\ell \in \L} V_\ell(\lambda)$, a linear combination of independent chi-squares; here, $V_\ell(\lambda) \equiv (\lambda e_\ell + 1)V_\ell$.  Next, let $\hat\lambda(\cdot)$ be the function that defines maximum likelihood estimator of $\lambda$ based on observations from the distribution of $V_{-\L}(\lambda)$; like in the R software, the negative subscript means those indices are removed.  Define
\begin{equation}
\label{eq:t.vc}
T_{Y,\psi} = H\Bigl( \frac{1}{\psi} \sum_{\ell \in \L} S_\ell \bigmid \hat\lambda(S_{-\L}/\psi) \Bigr)
\end{equation}
and
\begin{equation}
\label{eq:z.vc}
Z = H\Bigl( \sum_{\ell \in \L} V_\ell(\lambda) \bigmid \hat\lambda(V_{-\L}(\lambda)) \Bigr) 
\end{equation}
and consider the generalized association 
\begin{equation}
\label{eq:ga.vc}
T_{Y,\psi} = F_\lambda^{-1}(U), \quad U \sim \unif(0,1), 
\end{equation}
where $F_\lambda$ is the distribution function of $Z$ in \eqref{eq:z.vc}.  Note that if $\hat\lambda(V_{-\L}(\lambda))$ were exactly equal to $\lambda$, then $Z$ would be $\unif(0,1)$, and the problematic dependence on the nuisance parameter $\lambda$ would be eliminated.  However, it is too much to expect that $\hat\lambda(\cdot)$ will exactly equal $\lambda$, so the dependence on $\lambda$ remains, at least for small samples.  For the generalized association \eqref{eq:ga.vc}, an appropriate predictive random set for $U$ is the ``default'' $\S$ used above.  Then the construction of the generalized IM for $(\psi,\lambda)$ is straightforward.  Elimination of $\lambda$ will be carried out by optimizing over $\lambda$ as discussed in Section~\ref{SS:nuisance}.  

For illustration, I will revisit an example presented in \citet[][Section~4.1]{burch.iyer.1997} and \citet[][Section~5.2]{e.hannig.iyer.2008}, where $L=165$, the $e$'s range from $e_1=8.56$ to $e_L=0.57$, and each $r_\ell=1$ except $r_{105}=2$.  Following \citet{burch.iyer.1997}, I take $\L=\{82,\ldots,165\}$.  Figure~\ref{fig:cdf}(a) shows plots of the distribution function $F_\lambda$ for $\lambda \in \{0.1, 1, 10, 100\}$.  This shows that the distribution depends on $\lambda$, but maybe not too much.  A marginal plausibility interval for $\psi$, based on this generalized IM, can be obtained by setting $G(\psi) \equiv T_{Y,\psi}$ equal to each of the extreme 2.5\% quantiles---optimized over $\lambda$---and solving for the corresponding $\psi$; see, also, \citet{xiong.optim}.  A plot of $G(\psi)$ for these data is shown in Figure~\ref{fig:cdf}(b).  In this case, the 95\% plausibility interval for $\psi$ is $(0, 3.22)$, which is similar to, but shorter than, the fiducial interval given in \citet{e.hannig.iyer.2008}.  To check the claimed validity, 2000 independent data sets are simulated by plugging in the maximum likelihood estimator of $(\psi,\lambda)$.  The coverage probability of the 95\% marginal plausibility interval is 0.947 and the average length is 3.31.  

\begin{figure}
\begin{center}
\subfigure[Distribution function $F_\lambda$ of $Z$ in \eqref{eq:z.vc}]{\scalebox{0.6}{\includegraphics{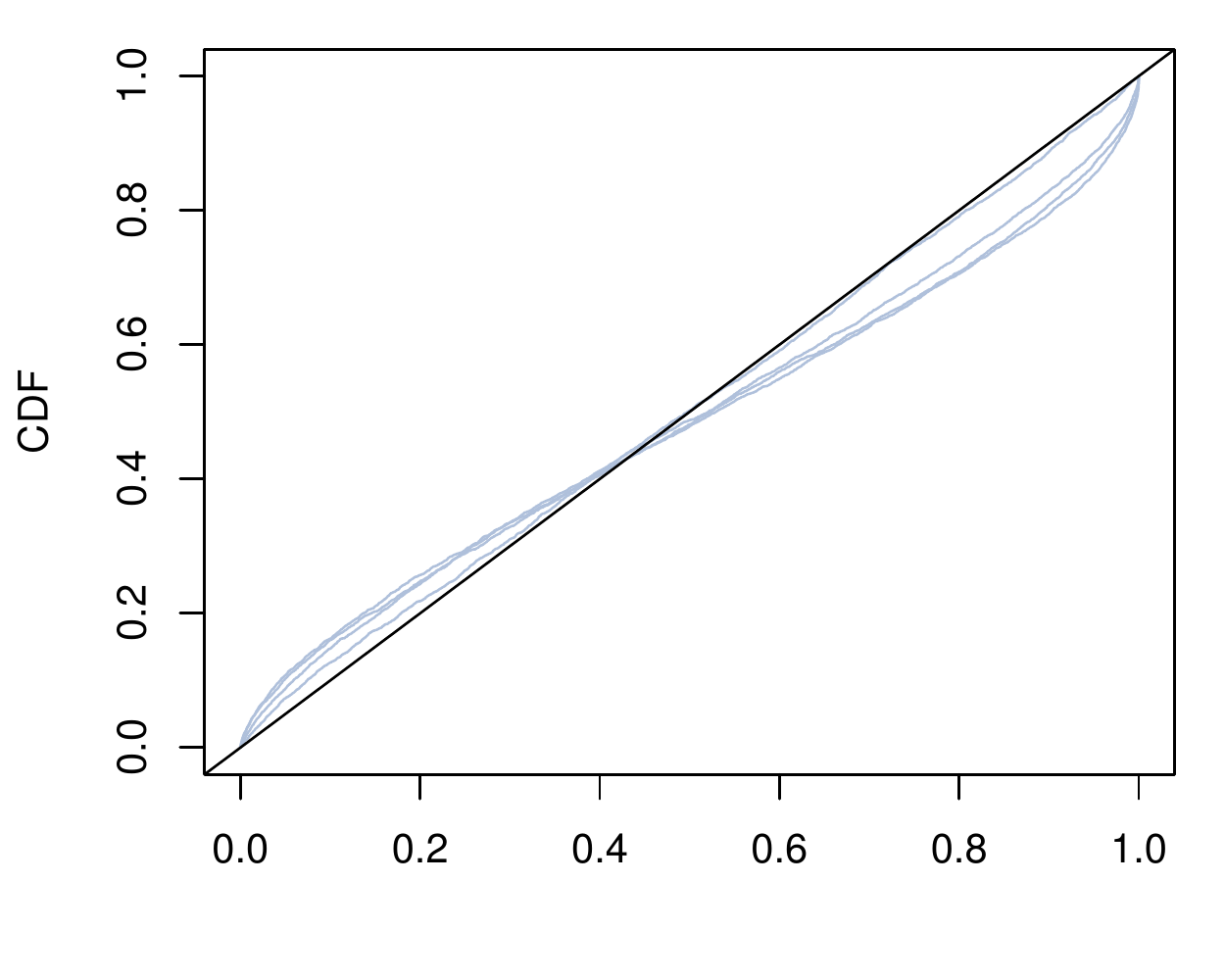}}}
\subfigure[Function $G(\psi) = T_{Y,\psi}$ in \eqref{eq:t.vc}]{\scalebox{0.6}{\includegraphics{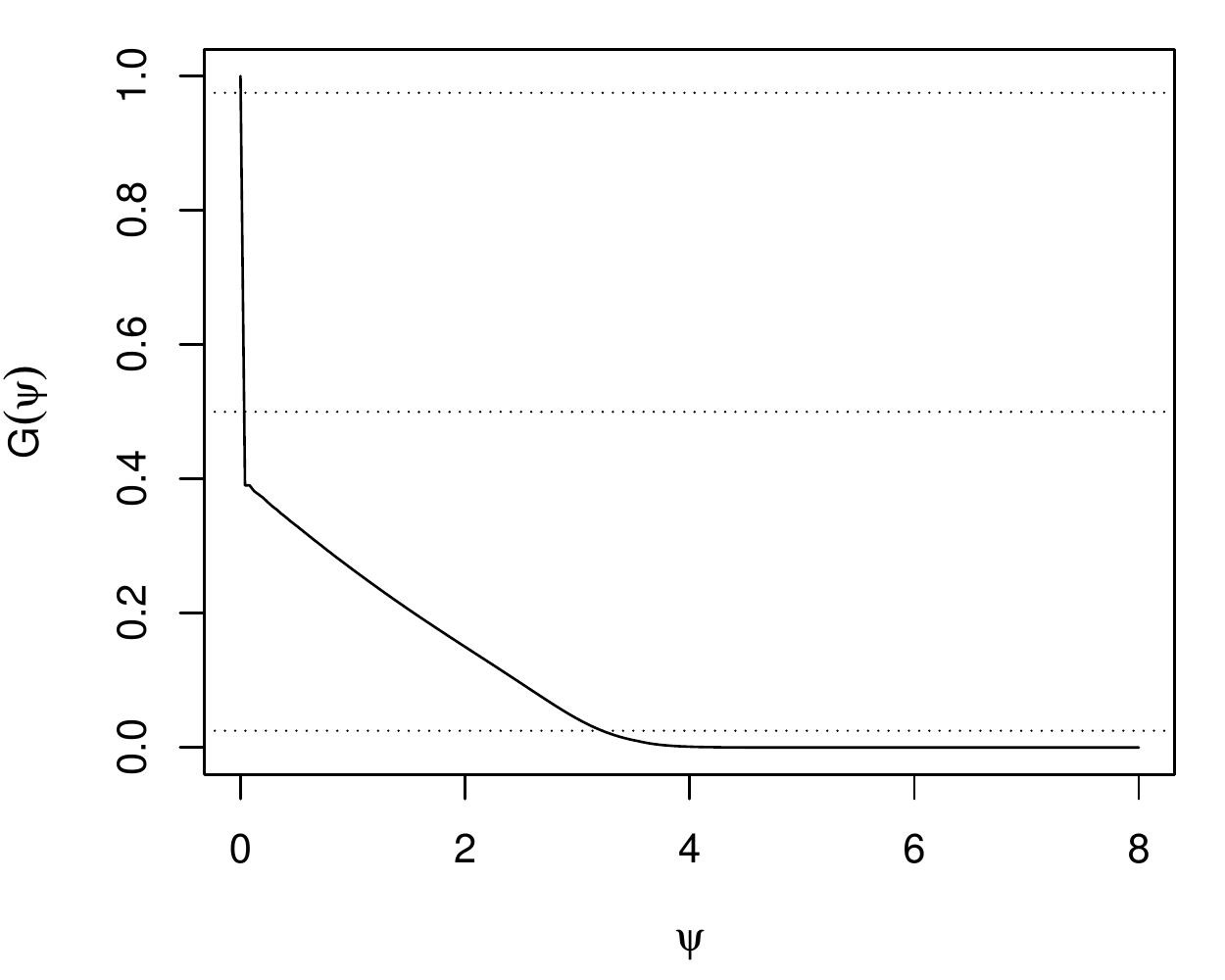}}}
\end{center}
\caption{Functions related to \eqref{eq:t.vc} and \eqref{eq:z.vc}, with $\L=\{82,\ldots,165\}$ and the data in \citet[][Section~4.1]{burch.iyer.1997}; in Panel~(a), $\lambda$ ranges over $\{0.1, 1, 10, 100\}$.}
\label{fig:cdf}
\end{figure}

The fiducial interval being compared to is of high quality \citep{e.hannig.iyer.2008}, so the fact that this generalized IM approach is competitive is quite promising.  Theoretically, the validity result holds, but computation is still a challenge.  For one thing, the method of \citet{imhof1961} used to evaluate $G(\psi)$, as implemented in the {\tt CompQuadForm} package in R, is a bit unstable when $\psi$ is close to zero.

\section{Discussion}
\label{S:discuss}

Previous work on IMs might give the impression that the approach is rigid in its dependence on a version of the data-generating process and, overall, not user-friendly.  In this paper, I have proposed a generalized version of the IM framework that is more flexible in a variety of ways.  In particular, it makes the IM approach more accessible by seamlessly incorporating some of the more familiar ideas from classical statistics.  This added flexibility does not require a sacrifice in terms of the IM's general validity property and, moreover, at least in certain cases, it leads to improved efficiency.  

There are at least two important questions that remain to be addressed.  First, what is an ``optimal'' choice of the mapping $T_{y,\theta}$?  Some simple ideas were presented in Remark~\ref{re:efficiency} but more work is needed.  Second, does this proposed strategy that collapses the problem down to one involving a scalar auxiliary variable work well even in high-dimensional problems?  It is likely that this extreme of dimension-reduction will result in a loss of efficiency when the problem is sufficiently complex, but this has yet to be investigated.

\section*{Acknowledgments}

The author thanks the Guest Editors, Professors Nils Hjort and Tore Schweder, for the opportunity to submit a paper for this special \emph{JSPI} issue, and to Professor Chuanhai Liu and the anonymous reviewers for some very helpful comments on a previous draft.  


\ifthenelse{1=1}{}{
\appendix

\section{Towards generalized IM efficiency}
\label{SS:score}

A notion of optimal predictive random sets was put forth in \citet{imbasics}.  At a high level, subject to a validity constraint, a predictive random set is optimal if it makes the plausibility function as stochastically small as possible.  For concreteness, take a fixed $\theta_0$ and let $\pl_y(\theta_0)$ denote the plausibility function evaluated at the singleton assertion $A=\{\theta_0\}$ based on the generalized IM in Section~\ref{SS:construction}.  \citet{imbasics} suggest to choose a valid (and nested) random set $\S$ such that 
\[ \psi_\alpha(\theta) = \prob_{Y|\theta}\{\pl_Y(\theta_0) > \alpha\} \]
takes its maximum value, $1-\alpha$, at $\theta=\theta_0$ for each $\alpha$.  This condition forces validity at $\theta=\theta_0$ but, for $\theta\neq\theta_0$, the plausibility function will tend to take smaller values, thereby also making the corresponding plausibility region small.  In the present context, these efficiency considerations will provide some clues about what makes a good choice of $g_\theta(y)$ in the generalized association \eqref{eq:ga}.  

Under the present setup, it is clear that there is a nested sequence $\YY_\alpha \subset \YY$, depending on $\S$ and $\theta_0$ such that $\{\pl_y(\theta_0; \S) > \alpha\}$ if and only if $y \in \YY_\alpha$.  Therefore, maximizing $\psi_\alpha(\theta)$ is equivalent to maximizing $\int_{\YY_\alpha} p_\theta(y) \,dy$, where $p_\theta$ is the density of $\prob_{Y|\theta}$.  Setting the first derivative equal to zero, under standard regularity conditions, one gets that 
\[ \int_{\YY_\alpha} S_\theta(y) \, p_\theta(y) \,dy = 0 \quad \text{at $\theta=\theta_0$ for all $\alpha$}, \]
where $S_\theta(y) = (\partial/\partial\theta) \log p_\theta(y)$ is the familiar score function.  Since $\E_{\theta_0}\{S_{\theta_0}(Y)\} = 0$, this condition shows that the set $\YY_\alpha$ needs to be suitably balanced with respect to the distribution of $S_{\theta_0}(Y)$.  Once $\YY_\alpha$ satisfying this \emph{score-balance condition} is identified, it can be traced back to find the corresponding mapping $g_\theta(y)$, etc.  

Here we justify the claim that that $\psi_\alpha(\theta)$ takes its maximum value under the score-balance condition.  To see this, differentiate $\psi_\alpha(\theta)$ at $\theta=\theta_0$ under the integral a second time to get 
\[ -\int_{\YY_\alpha} \{J_\theta(y) - S_\theta(y) S_\theta(y)^\top\} p_{\theta_0}(y) \,dy, \]
where $J_\theta(y)$ is the observed Fisher information; assume, for simplicity, that $J_{\theta_0}(Y)$ is positive definite with $\prob_{Y|\theta_0}$-probability~1.  Under the score-balance condition, if this second derivative matrix were negative definite, then $\psi_\alpha(\theta)$ would be (locally) maximized at $\theta=\theta_0$, as desired.  It suffices to show that 
\[ \int_{\YY_\alpha} \{I_d - J_{\theta_0}(y)^{-1/2} S_{\theta_0}(y) S_{\theta_0}(y)^\top J_{\theta_0}(y)^{-1/2}\} p_{\theta_0}(y) \,dy \quad \text{is positive definite}. \]
By Weyl's eigenvalue theorem {\color{blue} [ref?]}, it is enough to show that 
\[ \int_{\YY_\alpha} S_{\theta_0}(y)^\top J_{\theta_0}(y)^{-1} S_{\theta_0}(y) p_{\theta_0}(y) \,dy < 1. \]
{\color{blue}The integral of the quadratic form over all of $\YY$ is (approximately) 1, so restricting the integration to $\YY_\alpha$, where the quadratic form is small, gives a value less than 1.  Therefore, $\psi_\alpha(\cdot)$ is locally maximized at $\theta_0$, so the shape of $\YY_\alpha$ above is near optimal.  



Note that, given $\YY_\alpha$, it is not necessary to work backwards to identify the corresponding predictive random set---the calculations can be done directly in terms of the score function.  Indeed, the plausibility function in this case is 
\[ \pl_y(\theta_0) = \prob_{Y|\theta_0}\{S_{\theta_0}(Y)^\top J_{\theta_0}(Y)^{-1}S_{\theta_0}(Y) \geq S_{\theta_0}(y)^\top J_{\theta_0}(y)^{-1} S_{\theta_0}(y)\}. \]
The plausibility function above corresponds to the p-value for Rao's score test of $H_0: \theta=\theta_0$ \citep{impval} and, since this test is asymptotically efficient under certain regularity, one can conclude that the predictive random set constructed here is at least approximately optimal.  
}

{\color{red} Here we started from just a model and derived the plausibility function directly, i.e., no mention of a data-generating mechanism or a predictive random set is needed, and we still get a valid and, in some sense, optimal IM...!}

{\color{red} Example: bivariate standard normal correlation, with ``optimal score-balanced'' association...} 
}

\bibliographystyle{apalike}
\bibliography{/Users/rgmartin/Dropbox/Research/mybib}

\end{document}